\documentclass[runningheads]{llncs}
\usepackage{amssymb}
\usepackage{graphicx}
\usepackage{amsmath}
\usepackage{algorithm,algorithmic}
\usepackage{url}
\usepackage{lscape}
\usepackage{color}

\urldef{\mailsa}\path|ma_asyraf@upm.edu.my|
\urldef{\mailsb}\path|rezal@upm.edu.my|

\title{Rabin-$p$ Cryptosystem: Practical and Efficient Method for Rabin based Encryption Scheme}
\titlerunning{Rabin-$p$ Cryptosystem}
\author{M.A. Asbullah \inst{1} \and M.R.K. Ariffin\inst{1,2}} 
\date{}
\institute {Al-Kindi Cryptography Research Laboratory,\\
Institute for Mathematical Research \and
Department of Mathematics, Faculty of Science,\\
Universiti Putra Malaysia (UPM), Selangor, Malaysia\\
\mailsa, \mailsb\\}
\begin{document}
\maketitle
\date{}

\begin{abstract}
In this work, we introduce a new, efficient and practical scheme based on the Rabin cryptosystem without using the Jacobi symbol, message redundancy technique or the needs of extra bits in order to specify the correct plaintext. Our system involves only a single prime number as the decryption key and does only one modular exponentiation. Consequently, this will practically reduce the computational efforts during decryption process. We demonstrate that the decryption is unique and proven to be equivalent to factoring.The scheme is performs better when compared to a number of Rabin cryptosystem variants. 
\end{abstract}

\section{Introduction}

\subsection{Background} 
Prior to 1970’s, encryption and decryption were done only in symmetrical ways. This was the practice until the advent of public key cryptosystem that was introduced by [30]. Yet, at that time the notion of asymmetric cryptosystem is somehow not well realized by many people. In 1978 the RSA cryptosystem [24] went public and it is regarded now by the cryptographic community as the first practical realization of the public key cryptosystem. The security of RSA is based on the intractability to solve the modular $e$-th root problem coupled with the integer factorization problem (IFP) of the form $N=pq$.
\ 
In 1979, one year after the invention of RSA cryptosystem, Michael O. Rabin [20] introduced another cryptosystem which was based on the intractability to solve the square root modulo problem of a composite integer. In fact, this cryptosystem is the first public key system of its kind that was proved equivalent to factoring $N=pq$. At first glance, we might consider Rabin cryptosystem as an RSA variant with the use of the public exponent $e=2$ differing from RSA which utilizes public exponents $e\geq3$. Interestingly, this claim is not necessarily true since by definition $\gcd(e,\phi(N))=1$ for RSA, but $\gcd(e=2,\phi(N))\neq1$ in the case of Rabin’s exponent, where $\phi(N)=(p-1)(q-1)$. In addition, the role of the public exponent $e=2$ of Rabin encryption gives a computational advantage over the RSA.
\
In [20], the encryption is computed by performing a single squaring modulo $N$. This is far more efficient by comparison to the RSA encryption, which requires the calculation of at least a cubic modulo $N$ [3]. Based on recent results in this area the public exponent for RSA must be sufficiently large. Values such as $e=3$ (the smallest possible encryption exponent for RSA) and $e=17$ can no longer be recommended, but commonly used values such as $e=2^{16}+1=65537$ still serve to be fine, thus Rabin has some advantage regarding to this matter [1]. On the other hand, Rabin decryption breaks up into two parts. Firstly are the calculation of two modular exponentiations, and then the computation using Chinese Remainder Theorem (CRT) is used. Here the efficiency of the Rabin decryption is slightly faster than to RSA.
\\ 
\indent The Rabin encryption function in the form $c=m^2 \pmod N$, where $N=pq$ such that $p,q$ are primes congruence $3 \pmod 4$ is considered to be as hard as factoring problem. In other words, it is mathematically proven that a random plaintext can be recovered completely from the ciphertext, if and only if the adversary is able to efficiently factoring the public key $N=pq$. On the contrary, the RSA encryption in the form $c=m^e  \pmod N$ might be easier than factoring problem. This is the case because the equivalent of RSA encryption function vis-a-vis factoring is not yet proven. Therefore the process of finding the $e$-th root is maybe possible without the need of factoring $N=pq$. The security of the RSA encryption scheme is based on the assumption that modular $e$-th root problem is a one-way function and yet the only methods to find the $e$-th root might be only if the adversary is capable to efficiently factoring the public key $N=pq$.
\\
\indent In principle, the Rabin scheme is really efficient, because only a square is required for encryption; furthermore, it is shown to be as hard as factoring problem. Alas, the Rabin cryptosystem suffer from two major drawbacks; the foremost one is because the Rabin’s decryption produces four possible candidates, thus introduces ambiguity or unclearness to decide the correct message out of four possible values. Another drawback is from the fact that its equivalence relation to factorization. On one side, the Rabin cryptosystem gives confidence as the security of breaking such system is proven to be as difficult as factoring compare to RSA. On the other side, the computational equivalence relation of the Rabin cryptosystem and the integer factorization problem makes the scheme vulnerable to an adversary that can launch a stronger attack, namely the chosen ciphertext attack. In summation, any scheme that inherits the properties of a security reduction that is equivalent to factoring is not very practicable as cipher systems [26]. These two disadvantages of the Rabin encryption scheme prevented it from widespread practical use.
\

\subsection{Related work}
In spite of the situation of four-to-one mapping of Rabin’s decryption, and the vulnerability to chosen ciphertext attacks, several attempts were made to solve this problem adequately. It is very interesting to witness continuous efforts in searching for apractical and optimal Rabin cryptosystem by numerous scholars. We put forward the summary for Rabin’s variants as follows.
\

Williams [11] proposed an implementation of the Rabin cryptosystem using the Jacobi symbol. Williams solves for unique decryption while maintaining the property of Rabin’s scheme, i.e. breaking is equivalence to factoring. Subsequently, the utility of the Jacobi symbol as extra information to define the correct square root accompanied with Rabin cryptosystem was also proposed by [13]. Nonetheless, both encrypt-decrypt processes require the Jacobi symbol calculation.This consequence in turn goes to additional computational cost and the extra bits lead to ciphertext overhead.Next is the extra bits methodology. It is a very attractive approach to solve the uniqueness problem in the Rabin decryption procedure. It appeared in current literature such as extra bits introduced by [18] and also utilized together with the Dedikind Sums Theorem in [7].
\

In [3] a redundancy to the message was proposed, which intends to append the plaintext with the repeating of least significant bits of the message with a pre-defined length. This approach will likely provide the decryption with a unique output but has a small probability that decryption may fail. In [19, 29] and [15] the authors proposed a Rabin-type cryptosystem with alternative modulus choice of $N=p^2q$. Such moduli is claimed to be no easier as to factoring the conventional modulus of $N=pq$ [10]. The combination of Rabin cryptosystem with a specific padding method was proposed by [14], [19] by using Optimal Asymmetric Encryption Padding (OAEP) [17] and Rabin-SAEP [5]. Note that the message output by decryption process for this padding scheme is unique but the decryption may fail with small probability.

\subsection{Motivation and Contributions}

It is of practical considerations that motivated researchers to turn the Rabin scheme to be useful and practical as RSA since it possess practical qualities. In general, all the existing variation seems to apply some additional features, for instance; implementing some padding, adding redundancy in the message or manipulate some mathematical pattern, with the target to get a unique decryption result but at the same time losing its computational advantage over RSA. In order to engage this problem and to overcome all the mentioned shortcomings, further analytical work is needed to refine those existing work.
\

In this work, we revisit Rabin cryptosystem and propose anew efficient and practical scheme which has the following characteristics; 
\begin{enumerate}
	\item a cryptosystem that can be proven equivalent to factoring, 
	\item preserve the performance of Rabin encryption while producing a unique message after decryption, 
	\item improve decryption efficiency by using only one modular exponentiation as oppose to typical Rabin-based decryption that use two one modular exponentiation, 
	\item the decryption key using only a single prime number instead of two, 
	\item and finally, resilient to a side channel attack namely the Novak’s attack by avoiding the need for CRT computation.
\end{enumerate}

\subsection{Paper Organization}

Section 2 provides a list of drawbacks from previous strategies that need to be avoided for practically implementing the Rabin encryption scheme. In this section we also highlight the methodology of the research performed.In section 3, we describe our proposal, and attach to it with a numerical example. This is followed by rigorous analysis and discussion upon the security of the scheme, in Section 4. The conclusion appears in the final section.

\section{Preliminaries}

\subsection{Drawbacks of Previous Strategies}

In this section, we initiate a list that describes the drawback of the previous strategies to overcome the Rabin weaknesses. Hence we established conditions that need to be avoided on any attempt to refine the Rabin scheme.

\subsubsection{The Use of Jacobi Symbol}
The necessity to compute Jacobi symbol possibly during the key generation, encryption or decryption makes the scheme less efficient [5]. In terms of computational performance, basically Rabin encryption is extremely fast as long as this process does not require for computing a Jacobi symbol [25].

\subsubsection{Message Redundancy and Padding Mechanism}
Some schemes introduce redundancy upon the message or design padding mechanism aiming to achieve an efficient way to determine the correct plaintext from its four possible candidates. For instance, as for HIME(R) cryptosystem that applies OAEP and Rabin-SAEP with a padding mechanism that designated to be simpler than OAEP. However, both methods still have a probability of decryption failure – even though small.

\subsubsection{The Use of CRT}
We observe that Rabin’s entire variant involves the process of finding all the four square roots modulo a composite number.Such procedure, utilizes the CRT, hence the power analysis, side channel attack (see [23]) is applicable on such computation.

\subsection{Methodology}
Now, we outlined the methodology to overcome the drawbacks of the Rabin cryptosystem and all its variants. Firstly, we put the condition on the modulus of type $N=p^2q$ to be used. We then impose restriction on the plaintext $m$ and ciphertext $c$ space as $m\in \mathbb{Z}_{p^2}$ and $c \in \mathbb{Z}_{p^2q}$, respectively. From the plaintext-ciphertext expansion such restriction leads to a system that is not a length-preserving for the message.
\

Let $m$ and $c$ be the plaintext and ciphertext and $c(m)$ be the function of $c$ taking $m$ as its input. Say, for instance, in RSA the plaintext and ciphertext spaces are the same, thus we denote the mapping as $c(m): \mathbb{Z}_{pq}\longrightarrow \mathbb{Z}_{pq}$. Note that this situation could be an advantage for the RSA scheme since RSA encryption has no message expansion. This is not, however true for all cryptosystems. For example, the plaintext-ciphertext mapping for Okamoto-Uchiyama Cryptosystem [28] is $c(m):\mathbb{Z}_{p}\longrightarrow \mathbb{Z}_{p^2}$, Pailier cryptosystem [22] and the scheme by [8] is $c(m): \mathbb{Z}_{pq}\longrightarrow \mathbb{Z}_{(pq)^2}$, Rabin-SAEP [5] mapping is $c(m): \mathbb{Z}_{\frac{pq}{2}}\longrightarrow \mathbb{Z}_{pq}$ and Schmidt-Samoa cryptosystem [15] is $c(m): \mathbb{Z}_{pq}\longrightarrow \mathbb{Z}_{p^2q}$.
\

The maximum space size is determined by the plaintext space. One way to do it would be to tell the user a maximum number of bits for the plaintext messages. If we view the message as merely the keys for a symmetric encryption scheme, meaning that the message is indeed a short message, then this is fine as many others schemes also implement this approach. Thus, we argue that the restriction of message space would be a hindrance is not an issue. 

\section{Our Proposed Scheme: The Rabin-$p$ Cryptosystem} 

In this section, we provide the details of the proposed cryptosystem namely Rabin-$p$ Cryptosystem. Rabin-$p$ is named after the Rabin cryptosystem with the additional $p$ symbolizing that the proposed scheme only uses a single prime $p$ as the decryption key. The proposed cryptosystem defined as follows.\\
\\
\begin{tabular}{l}
Algorithm 1. \textbf{Key Generation}\\
\hline
INPUT: The size $k$ of the security parameter.\\
OUTPUT: The public key $N=p^2q$ and the private key $p$. \\
\hline
1. Generate two random and distinct primes $p$ and $q$ \\ \ \ \ \ satisfying $p\equiv 3 \pmod 4$ and $q\equiv 3 \pmod 4$ \\
2. Compute $N= p^2q$. \\
3. Return the public key $N$ and the private key $p$.\\
\hline
\end{tabular}
\\ 
\\ \\
\begin{tabular}{l}
Algorithm 2. \textbf{Encryption}\\
\hline
INPUT: The plaintext $m$ and the public key $N$.\\
OUTPUT: A ciphertext $c$.\\
\hline
1. Choose plaintext $m<2^{2k-1}$ such that $\gcd(m,N)=1$.\\
2. Compute $c \equiv m^2 \pmod N$.\\
3. Return the ciphertext $c$.\\
\hline
\end{tabular}

\begin{remark} \label{remark1}
We observe that the message space is restricted to the range $m<2^{2k-1}$. It shows that the message $m<2^{2k-1}=\frac{2^{2k}}{2}<\frac{p^2}{2}<p^2$.
\end{remark}
\
\begin{tabular}{l}
Algorithm 3. \textbf{Decryption}\\
\hline
INPUT: A ciphertext $c$ and the private key $p$.\\
OUTPUT: The plaintext $m$\\
\hline
1. Compute $w\equiv c \pmod p$.\\
2. Compute $m_p\equiv w^{\frac{p+1}{4}} \pmod p$.\\
3. Compute $i= \frac{c-{m_p}^2}{p}$.\\
4. Compute $j\equiv \frac{i}{2m_p} \pmod p$.\\
5. Compute $m_1=m_p+jp$.\\
6. If $m_1<\frac{p^2}{2}$, then return $m=m_1$.\\
7. Else return $m=p^2-m_1$.\\
\hline
\end{tabular}
\\

\begin{remark} \label{remark2}
The decryption algorithm needs only a single prime number as its key and it operates with single modular exponentiation operation. This situation would give impact on the overall computational advantage of the proposed scheme against other Rabin variants.
\end{remark}

\subsection{\textbf{Proof of Correctness}}

Suppose $c\equiv m^2 \pmod {p^2q}$ then $c-m^2 \equiv 0 \pmod {p^2q}$. Note that, if $c-m^2$ is divisible by $p^2q$ then it is certainly divisible by $p^2$. From Remark \ref{remark1} the message $m<2^{2k-1}=\frac{2^{2k}}{2}<\frac{p^2}{2}<p^2$. Thus, it is sufficient to solve for $c\equiv m^2 \pmod {p^2}$ as stated in the following lemma.

\begin{lemma} \label{soln}
	Let $c \equiv m^2 \pmod {p^2}$ then there exist two distinct square roots of $c$; i.e. $m_1$ and $m_2$ such that $m_1\neq m_2$  and ${m_1}^2 \equiv {m_2}^2 \equiv c \pmod {p^2}$.
\end{lemma}

\begin{proof}
Suppose $m_1\neq m_2$ such that ${m_1}^2 \equiv {m_2}^2 \equiv c \pmod {p^2}$. Then
\begin{eqnarray}
{m_1}^2-{m_2}^2 \equiv (m_1+m_2)(m_1-m_2) \equiv 0 \pmod {p^2}
\end{eqnarray}
	
Note that $p^2|(m_1+m_2)(m_1-m_2)$, thus consider $p|(m_1+m_2 )(m_1-m_2 )$ as well. If $p|(m_1+m_2)$ and $p|(m_1-m_2)$, then $p$ would divide $(m_1+m_2 )+(m_1-m_2 )=2m_1$ and $(m_1+m_2 )-(m_1-m_2 )=2m_2$. Since $p \equiv 3 \pmod 4$ is an odd prime, then $p\nmid 2$ so $p$ would divide both $m_1$ and $m_2$. Consider ${m_1}^2 \equiv c \pmod {p^2}$ thus ${m_1}^2=c+lp^2$ for some integer $l$. If $p|m_1$ then $p|{m_1}^2$ therefore $p|c$. Observe that $\gcd(c,p)=1$ therefore $p\nmid c$. Hence $p\nmid m_1$. The same goes for $p\nmid m_2$. 

Now, consider in the case if $p|(m_1+m_2)$ or $p|(m_1-m_2)$ but not both. Since $p^2|(m_1+m_2 )(m_1-m_2 )$, therefore either $ p^2|(m_1+m_2)$ or $p^2|(m_1-m_2)$. This concludes $m_1 \equiv m_2 \pmod {p^2}$ and $m_1 \equiv -m_2 \pmod {p^2}$. \qed
\end{proof}

\begin{theorem}
If $m<2^{2k-1}$ such that $c \equiv m^2 \pmod N$ , then the Algorithm 3 (Decryption) will output the unique $m$.
\end{theorem}

\begin{proof}
From Lemma \ref{soln}, there are two distinct solution for $c \equiv m^2 \pmod N$. Now, we show that the decryption output by the Algorithm 3 is correct and produces only a unique solution for $m<2^{2k-1}$. We shall break down the proof in two separate lemmas as follows.
\end{proof}

\begin{lemma}[Correctness] \label{correct}
Suppose $m_p$ is a solution to $w \equiv c \pmod p$. Let $j$ be a solution to $2m_pj \equiv i \pmod p$ such that $i= \frac{c-{m_p}^2}{p}$. Then $m=m_p+jp$ is a solution to $c \equiv m^2 \pmod {p^2}$. Furthermore $-m \pmod {p^2}$ is another solution.
\end{lemma}

\begin{proof}
Suppose we are given the ciphertext $c$ as described in the encryption algorithm, and need to solve for its square root modulo $N=p^2q$. Let $c \equiv m^2 \pmod N$, and since $p|N$, then we have $w \equiv m^2 \equiv c \pmod p$. From here, since $m<p^2$ thus it is sufficient just solving $w\equiv c \pmod p$.
We begin by solving $w \equiv c \pmod p$. Let $m_p$ is a solution to $w \equiv c \pmod p$ such that $w \equiv {m_p}^2 \pmod p$. It thus suffices to find for $m_p$ the values $m=m_p+jp$ for some integer $j$ that we will find later. Suppose that  $m=m_p+jp$ is a solution for $w \equiv c \pmod p$, then we have
\begin{eqnarray}
c \equiv m^2 \equiv (m_p+jp)^2 \equiv {m_p}^2+2m_p jp \pmod {p^2}
\end{eqnarray}

So, the above congruence can be rearranged as $2m_pjp  \equiv c-{m_p}^2  \pmod {p^2}$. Note that from $w \equiv {m_p}^2 \pmod p$, we have $c-{m_p}^2 \equiv 0 \pmod p$ which means that $c-{m_p}^2$ is a multiple of $p$, say $ip$ for some integer $i$. From here, we could simply compute $i=\frac{c-{m_p}^2}{p}$. We then rewrite this equation as
\begin{eqnarray}
2m_pjp  \equiv ip \pmod {p^2}	
\end{eqnarray}

Hence, $p$ factor immediately cancelled out from $2m_pjp \equiv ip \pmod {p^2}$ since it implies that $2m_pj \equiv i \pmod p$. Hence, we compute $j \equiv \frac{i}{2m_p} \pmod p$. 
To conclude, we have a solution $m=m_p+jp$ for $c\equiv m^2 \pmod {p^2}$. Observe that $-m \pmod {p^2}$ is also another solution. \qed
\end{proof}	
\begin{corollary}
Let $m_1$ and $m_2$ be the solution for $c \equiv m^2 \pmod {p^2}$ then $m_1+m_2=p^2$.
\end{corollary}

\begin{proof}
Suppose $\pm m \pmod {p^2}$ are the solution for $c\equiv m^2 \pmod {p^2}$ as shown by Lemma \ref{correct}. Observe that if we set $m_1\equiv m \pmod {p^2}$ then $m_2 \equiv -m \pmod {p^2}$. Thus, $m_2$ can be simply written as $p^2-m$. Conclusively, $m_1+m_2=p^2$. \qed
\end{proof}
\noindent Now, the following lemma shows that the decryption algorithm will output the unique solution $m$.	

\begin{lemma}[Uniqueness]
Let $m<2^{2k-1}$. Then the decryption algorithm will output the unique $m$.
\end{lemma}

\begin{proof}
Observe that the upper bound for $m$ is $2^{2k-1}-1< \frac{p^2}{2}$. Consider $m_1+m_2=p^2$ with $p^2$ is an odd integer. Then either $m_1$ or $m_2$ is less than $\frac{p^2}{2}$ that satisfies the upper bound of $m<2^{2k-1}$. Observe that $p^2$ is an odd integer, then by definition $\frac{p^2}{2}$ is not an integer. Since that $m_1$ and $m_2$ need to be integers, thus $m_1,m_2\neq \frac{p^2}{2}$.
\

Suppose we consider both $m_1$ and $m_2$ are less  $\frac{p^2}{2}$, then we should have $m_1+m_2<p^2$ therefore we have a contradiction (i.e the fact that $m_1+m_2=p^2$). On the other hand, if we consider both $m_1$ and $m_2$ are greater than $\frac{p^2}{2}$, then we should have $m_1+m_2>p^2$ yet we reach the same contradictory statement.Thus, one of $m_1$ or $m_2$ must be less than $\frac{p^2}{2}$.
\
 
Suppose $m_1<\frac{p^2}{2}$ then there must exist a real number $\epsilon_1$ such that $m_1  + \epsilon_1  =  \frac{p^2}{2}$. On the other site, since we let $m_1<\frac{p^2}{2}$, then $m_2$ must be greater than $\frac{p^2}{2}$. Suppose $m_2>\frac{p^2}{2}$ then there must exist a real number $\epsilon_2$ such that $m_2 - \epsilon_2  =  \frac{p^2}{2}$. If we add up these two equations, we should have
\begin{eqnarray}
(m_1 + \epsilon_1)+(m_2 - \epsilon_2)= \frac{p^2}{2} + \frac{p^2}{2}=p^2	
\end{eqnarray}

But since we have $m_1+m_2=p^2$, thus $(\epsilon_1- \epsilon_2 )$ should be equal to zero, meaning that $\epsilon_1=\epsilon_2$. Finally, we conclude that only one of $m_1$ or $m_2$ are less than $\frac{p^2}{2}$ and will be outputted by the decryption algorithm as the uniquem. \qed
\end{proof}

\section{Analysis and Discussion}

\subsection{Equivalent to Factoring $N=p^2q$}
If a new cryptosystem is designed then we are expected to provide a comparison of the relative difficulty of breaking the scheme to the solving of any existing hard problems. In this case is breaking the proposed scheme is reducible to factoring the modulus $N=p^2q$? We may put this as if there exists an algorithm to factor the modulus $N=p^2q$ then there exists an algorithm to find the message $m$ from the ciphertext $c$ output by the proposed scheme. The converse is also true, as will be exemplified in the following theorem.

\begin{theorem} \label{factoring}
Breaking the Rabin-$p$ cryptosystem scheme is equivalence to factoring the modulus $N=p^2q$.
\end{theorem}

\begin{proof}
We show the proof for both directions as follows.	
\\

\noindent ($\Rightarrow$) Breaking the Rabin-$p$ cryptosystem is reducible to factoring $N=p^2q$.

\
Suppose we have an algorithm with the ability to factor the modulus $N=p^2q$, then we can solve the message $m$ from the ciphertext $c$ output by the proposed scheme simply by following the outlined decryption algorithm. Therefore the proposed scheme is reducible to factoring.
\\ \\
\noindent ($\Leftarrow$) Factoring $N=p^2q$ is reducible to breaking the Rabin-$p$ cryptosystem.

\
Conversely, suppose there exists an algorithm that break the proposed scheme; that is able to find the message $m$ from the ciphertext $c$ then there exists an algorithm to solve the factorization of the modulus $N=p^2q$. Implying that someone who can decrypt the message $m$ from the ciphertext $c$ must also be able to factor $N=p^2q$. \qed
\end{proof}

\noindent The factoring algorithm is defined as follows.
\\ \\
\begin{tabular}{l}
Algorithm 4. \textbf{Factoring Algorithm}\\
\hline
INPUT: A ciphertext $c$ and the modulus $N$.\\
OUTPUT: The prime factors $p^2$,$q$.\\
\hline
1. Choose an integer $2^{2k-1}<\hat{m}<2^{2k}-1$ \\
2. Compute $\hat{c} \equiv \hat{m}^2 \pmod N$.\\
3. Ask the decryption of ciphertext $\hat{c}$\\
4. Receive the output $m^\prime<2^{2k-1}$.\\
5. Compute $\gcd(\hat{m}\pm m^\prime, N)= p^2$.\\
6. Compute $\frac{N}{p^2}=q$.\\
7. Return the prime factors $p^2,q$.\\
\hline
\end{tabular}
\\
\subsection{Coppersmith's Theorem}
Coppersmith [6] introduced a significantly powerful theorem for ﬁnding small roots of modular polynomial equations using the LLL algorithm. In general, finding solutions to modular equations is easy if we know the factorization of the modulus. Else, it can be hard. When working with modulo a prime, there is no reason to use Coppersmith’s theorem since there exist far better root-finding algorithm (i.e Newton method). Since then, this method has found many different applications in the area of public key cryptography. Interestingly, besides many cryptanalytic results it has likewise been employed in the design of provably secure cryptosystems. The method has found many different applications in the area of public key cryptography, for example, such as to factor $N=pq$ if the high bits of $p$ are known, and attacking stereotyped messages in RSA by sending messages to whose differences is less than $N^{1/e}$. Now we provide the Coppersmith's theorem.

\begin{theorem}\textbf{[6]} \label{copp1}
Let $N$ be an integer of unknown factorization. Let $f_N(x)$ be a univariate, a monic polynomial of degree $\delta$. Then we can find all solutions $x_0$ for the equation $f_N(x) \equiv 0 \pmod N$ with $|x_0|< N^{1/\delta}$ in polynomial time.
\end{theorem}

\begin{corollary}
Let $c \equiv m^2 \pmod N$ from the ciphertext. If $m<2^{3n/2}$ then $m$ can be found in polynomial time.
\end{corollary}

\begin{proof}
Suppose $c\equiv m^2 \pmod N$. Consider the univariate, monic polynomial $f_N(x)\equiv x^2-c \equiv 0 \pmod N$, hence $\delta=2$.Thus, by applying Theorem \ref{copp1} the root $x_0=m$ can be recovered if $m<N^{1/\delta}=N^{1/2}\approx 2^{3n/2}$. Therefore, to avoid this attack, we need to set $m>2^{3n/2}$. \qed
\end{proof}

\begin{theorem} \textbf{[2]} \label{copp2}
Let $N$ be an integer of unknown factorization, which has a divisor $b>N^\beta$ . Furthermore, let $f_b(x)$ be a univariate, a monic polynomial of degree $\delta$.Then we can find all solutions $x_0$ for the equation $f_b(x)\equiv 0 \pmod b$ with $|x_0|<\frac{1}{2}N^{\beta^2/\delta}$ in polynomial time.
\end{theorem}

\begin{corollary}
Let $c \equiv m^2 \pmod {p^2}$ such that $p^2$ is an unknown factor for $N$. If $m<2^{2n/3}$ then $m$ can be found in polynomial time.
\end{corollary}

\begin{proof}
Suppose $c\equiv m^2 \pmod {p^2}$ such that $p^2$ is an unknown factor for $N$. Consider $f_{p^2}(x) \equiv x^2-c \equiv 0 \pmod {p^2}$ with $p^2\approx N^{2/3}\approx 2^{2n}$. We can find a solution $x_0=m$ if $m< \frac{1}{2}N^{\beta^2/\delta}< N^{(2/3)^2/2}= N^{2/9} \approx 2^{2n/3}$. \qed
\end{proof}

\subsection{Chosen Ciphertext Attack}

Notice that the factoring algorithm mentioned by the Algorithm 4 could provide a way to launch a chosen ciphertext attack upon the proposed scheme in polynomial time, hence resulting in the system totally insecure in this sense. Therefore, in order to provide security against this kind of attack, we could consider implementing any hybrid technique with symmetric encryption. The result from [9] is suitable for our scheme in order to achieve chosen ciphertext security, with the cost of a hash function. We may also apply the chosen ciphertext secure hybrid encryption transformation that was proposed in [16].

\subsection{Side Channel Attack}
In general, the decryption algorithm of a Rabin based cryptosystem consists of two parts where the first part is the modular exponentiation operation in order to obtain the value of message modulo $p$ and $q$ from the ciphertext $c$. The second part then would be the recombination process using the CRT to recover the proper message $m$. Most side channel attacks deal with the first part. For instance, using a timing attack approach [21, 31 and 4] or attacking by power analysis approach [27]. Alternatively, Novak [23] proposed a side channel attack on the second part (i.e. CRT computation). The result of Novak’s attacks upon the implementation of CRT shows that for a certain characteristic function can be detected by power analysis.Hence the corresponding modulus can be factored using such characteristic function. We reason that since our proposed scheme does not need to carry out any CRT computation, thus the Novak attack is not applicable on the Rabin-$p$ cryptosystem.

\subsection{Comparison}
Table 1 illustrates a comparison between our proposed scheme and others Rabin variants. Observe that only [20] and [29] suffers from the four to one decryption drawback situation. The works [13, 14 and 18] provide unique decryption with the aid of extra bits; even so the adversary may obtain some advantage since these additional bits may leak some useful information. In addition, these schemes also use the Jacobi symbol, which we already stressed out that such implementation is indeed inefficient during encryption. Furthermore, the decryption of the other Rabin variants that applies the message redundancy or message padding methods (i.e. such as [3] or [5, 19], respectively) may fail with a small probability.
\

We claim that our proposed scheme produces a unique decryption without even using the Jacobi symbol, without sending extra bits of information, and without using any redundancy or padding onto the message. Note that as far as the Novak attacks on the computation of CRT is concerned; our proposed scheme and the method described in [11] are not susceptible. Hence, these two schemes are secure against the Novak attacks. For the computation of modular exponentiation, once again all the mentioned Rabin variants in this paper perform two modular exponentiations, except for [11] and our scheme. Nevertheless, all the exponent of the size $p$, except of the size $pq$ for [11], thus we gain an advantage if compared to such scheme.

\begin{center}
\begin{center}
Table 1. Comparison between Our Proposed Scheme \\ 
and the Other Rabin Variants.	
\end{center}
\begin{tabular}{ | l | p{2cm} | p{2cm} | p{2cm} | p{2cm} | p{2cm} | }
\hline 
Scheme & Unique Decryption & Jacobi Symbol & Special Message Structure & Novak's attacks on CRT & Modular Exponentiation in Decryption\\ \hline
[20] & No (4-to-1 mapping) &	No & No & Yes & 2 \\ \hline
[11] &	Yes &	Yes & Yes (Size Restriction) &	No & 1** \\ \hline
[13] &	Yes (With extra bits) &	Yes &	No & Yes & 2 \\ \hline
[3] &	Yes* &	No &	Yes (Redundancy) & Yes	& 2 \\ \hline
[29] &	No (4-to-1 mapping) &	No &	No &	Yes &	2 \\ \hline
[5] &	Yes* &	No &	Yes (Size Restriction, Padding) & Yes & 2 \\ \hline
[14] &	Yes (With extra bits) &	Yes &	Yes (Padding) &	Yes & 2 \\ \hline
[19] &	Yes* &	No &	Yes (Padding) &	Yes & 2 \\ \hline
[18] &	Yes (With extra bits) &	Yes (with Dedikind Sums) & No &	Yes & 2 \\ \hline
[7] &	Yes &	Yes &	No & Yes & 2 \\ \hline
Our scheme &	Yes &	No &	Yes (Size Restriction) & No & 1 \\ 
\hline 
\end{tabular}
* Decryption may fail with a small probability.\\
** All the exponent of the size $p$, except of the size $pq$ for [11].
\end{center}

\section{Conclusion}
Rabin-$p$ cryptosystem is purposely designed without using the Jacobi symbol, redundancy in the message and avoiding the demands of extra information for finding the correct plaintext. Decryption outputs a unique plaintext without any decryption failure. In addition, decryption only requires a single prime. Furthermore, the decryption procedure only computes a single modular exponentiation instead of two modular exponentiation executed by other Rabin variants. As a result, this reduces computational effort during decryption process. Some possible attacks such as Coppersmith’s technique, chosen ciphertext attack and side channel attack have been analyzed. Still, none can successfully affect the proposed strategy. Finally, we show that Rabin-$p$ cryptosystem is performs better when compared to a number of Rabin variants.


\begin{thebibliography}{99}

\bibitem{}	A. Lenstra, and E. R. Verheul, Selecting cryptographic key sizes, Journal of cryptology. 14(4) (2001), 255-293. 
\bibitem{}	A. May, New RSA vulnerabilities using lattice reduction methods,PhD diss., University of Paderborn, 2003.
\bibitem{}	A. Menezes, C. van Oorschot, and A. Vanstone,Handbook of applied cryptography, CRC Press, Washington, 1997.
\bibitem{}	D. Brumley, and D.Boneh, Remote timing attacks are practical, Computer Networks 48 (2005), 701-716.
\bibitem{}	D. Boneh, Simplified OAEP for the RSA and Rabin functions, Advances in Cryptology—CRYPTO 2001. Springer Berlin Heidelberg, 2001.
\bibitem{}	D. Coppersmith, Small solutions to polynomial equations, and low exponent RSA vulnerabilities, Journal of Cryptology. 10(4) (1997), 233-260. 
\bibitem{}	D. Freeman, O. Goldreich, E. Kiltz, A. Rosen, and G. Segev, More constructions of lossy and correlation-secure trapdoor functions, Journal of cryptology. 26(1) (2013), 39-74.
\bibitem{}	D. Galindo, S. Martýn, P. Morillo, and J. Villar, A practical public key cryptosystem from Paillier and Rabin schemes, Public Key Cryptography—PKC 2003. Springer Berlin Heidelberg, 2003. 
\bibitem{}	D. Hofheinz, and E. Kiltz, Secure hybrid encryption from weakened key encapsulation, Advances in Cryptology—CRYPTO 2007. Springer Berlin Heidelberg, 2007.
\bibitem{}	G. Castagnos, A. Joux, F. Laguillaumie, and P. Nguyen, Factoring $pq^2$ with quadratic forms: NICE cryptanalyses, Advances in Cryptology—ASIACRYPT 2009. Springer Berlin Heidelberg, 2009. 
\bibitem{}	H. C. Williams, A modification of the RSA public-key encryption procedure, IEEE Transactions on Information Theory. 26(6) (1980), 726-729.
\bibitem{}	K. Kurosawa, and Y. DesmedtA new paradigm of hybrid encryption scheme, Advances in Cryptology—Crypto 2004. Springer Berlin Heidelberg, 2004.
\bibitem{}	K. Kurosawa, T. Ito, and M. Takeuchi, Public key cryptosystem using a reciprocal number with the same intractability as factoring a large number, Cryptologia. 12(4) (1988), 225-233.
\bibitem{}	K. Kurosawa, W. Ogata, T. Matsuo, and S. Makishima, IND-CCA public key schemes equivalent to factoring n= pq, Public Key Cryptography—PKC 2001. Springer Berlin Heidelberg, 2001.
\bibitem{}	K. Schmidt-Samoa, A new Rabin-type trapdoor permutation equivalent to factoring, Electronic Notes in Theoretical Computer Science. 157(3) (2006), 79-94.
\bibitem{}	M. Abe, R. Gennaro, and K. Kurosawa, Tag-KEM/DEM: A new framework for hybrid encryption, Journal of Cryptology. 21(1) (2008), 97-130.
\bibitem{}	M. Bellare, and P. Rogaway, Optimal asymmetric encryption, Advances in Cryptology—EUROCRYPT'94. Springer Berlin Heidelberg, 1995.
\bibitem{}	M.Elia, M. Piva, and D. Schipani, The Rabin cryptosystem revisited. arXiv preprint (2011). Available at arXiv:1108.5935.
\bibitem{}	M. Nishioka, H. Satoh, and K. Sakurai, Design and analysis of fast provably secure public-key cryptosystems based on a modular squaring, Information Security and Cryptology—ICISC 2001. Springer Berlin Heidelberg, 2002. 
\bibitem{}	M. O. Rabin, Digitalized signatures and public-key functions as intractable as factorization, Technical Report (1979).
\bibitem{}	P. Kocher, Timing attacks on implementations of Diffie-Hellman, RSA, DSS, and other systems, Advances in Cryptology—CRYPTO’96. Springer Berlin Heidelberg, 1996.
\bibitem{}	P. Paillier, Public-key cryptosystems based on composite degree residuosity classes, Advances in cryptology—EUROCRYPT’99. Springer Berlin Heidelberg, 1999.
\bibitem{}	R. Novak, SPA-based adaptive chosen-ciphertext attack on RSA implementation, Public Key Cryptography—PKC2002. Springer Berlin Heidelberg, 2002.
\bibitem{}	R. Rivest, A. Shamir, and L. Adleman, A method for obtaining digital signatures and public-key cryptosystems, Communications of the ACM. 21(2) (1978), 120-126.
\bibitem{}	S. D. Galbraith, Mathematics of public key cryptography. Cambridge University Press, 2012.
\bibitem{}	S. Müller,On the security of Williams based public key encryption scheme, Public Key Cryptography—PKC2001. Springer Berlin Heidelberg, 2001.
\bibitem{}	T. Messerges, E. Dabbish and R. Sloan, Power analysis attacks of modular exponentiation in smartcards, Cryptographic Hardware and Embedded Systems. Springer Berlin Heidelberg, 1999.
\bibitem{}	T. Okamoto, and S. Uchiyama, A new public-key cryptosystem as secure as factoring, Advances in Cryptology—EUROCRYPT'98. Springer Berlin Heidelberg, 1998.
\bibitem{}	T. Takagi, Fast RSA-type cryptosystem modulo $p^k q$, Advances in Cryptology—CRYPTO'98. Springer Berlin Heidelberg, 1998.
\bibitem{}	W. Diffie, and M. Hellman, New directions in cryptography, IEEE Transactions onInformation Theory. 22(6) (1976), 644-654.
\bibitem{}	W. Schindler, A timing attack against RSA with the Chinese Remainder Theorem, Cryptographic Hardware and Embedded Systems—CHES 2000. Springer Berlin Heidelberg, 2000.
	
\end{thebibliography}
\end{document}